\documentclass{article}
\usepackage{amsfonts}
\usepackage{mathrsfs}
\usepackage{amsthm}
\usepackage{amssymb}
\usepackage{amsmath}
\usepackage{enumerate}
\usepackage{pb-diagram}
\usepackage[all]{xy}
\usepackage{color}

\theoremstyle{plain}
\newtheorem{theorem}{Theorem}[section]
\newtheorem{proposition}[theorem]{Proposition}
\newtheorem{corollary}[theorem]{Corollary}
\newtheorem{lemma}[theorem]{Lemma}

\theoremstyle{definition}
\newtheorem{definition}{Definition}[section]

\theoremstyle{remark}
\newtheorem{remark}{\textbf{Remark}}[section]

\theoremstyle{example}

\numberwithin{equation}{section}

\title{Continuous-Time Quantum Walk on Locally Infinite Graph}
\author{Ce Wang\\
Shanghai Institute for Mathematics and Interdisciplinary Sciences\\
Shanghai 200433, People's Republic of China}

\begin{document}
\maketitle

\noindent\textbf{Abstract.}\ \
Time-reversal symmetry is of fundamental importance to physics.
In the classical theory of time-reversal symmetry, the time-reversal symmetry of a quantum system
is described by an anti-unitary operator, which is known as the time-reversal operator of the system.
In this paper, we introduce and study a model of continuous-time quantum walk on a special locally infinite graph.
After examining its spectral property, we investigate the time-reversal symmetry of the model.
To our surprise, we find that its time-reversal symmetry can be described directly by a unitary operator,
which contrasts sharply with that in the classical theory of time-reversal symmetry. Some other related results are also proven.
\vskip 2mm

\vskip 2mm


\section{Introduction}

Time-reversal symmetry is of fundamental importance to physics \cite{geru}. Consequently, it arises naturally in quantum mechanics.
In the classical theory of time-reversal symmetry, the time-reversal symmetry in the evolution of a quantum system
is described by an anti-unitary operator, which is known as the time-reversal operator of the system \cite{geru}.
Unlike unitary operators, which are linear, anti-unitary operators are anti-linear (also known as conjugate linear).

In recent two decades, much attention has been paid to continuous-time quantum walks
(see, e.g. \cite{bose,childs,konno,lin,sin, W-IJQI-2023, zhan} and references therein),
which were originally introduced by Farhi and Gutmann \cite{farhi-gutmann} as quantum counterparts of the classical continuous-time random walks,
and now have wide application in quantum computation, quantum communication as well as in modeling physical processes
(see \cite{venegas} and references therein). From a viewpoint of mathematical physics, a continuous-time quantum walk can be viewed as a small quantum system consisting of a quantum particle (walker), and its evolution still obeys the Schr\"{o}dinger equation.
Hence, it is natural to analyze continuous-time quantum walks from a viewpoint of time-reversal symmetry.

Quantum Bernoulli noises \cite{WCL-JMP-2010} are annihilation and creation operators acting on Bernoulli functionals,
which satisfy the canonical anti-commutation relations (CAR) in equal time.
Due to the variety in their algebraic and analytical structures, quantum Bernoulli noises
have found wide application in many problems in mathematical physics (see, e.g. \cite{WTR-JMP-2019, WY-QIP-2016}).
In this paper, we introduce and study a model of continuous-time quantum walk based on quantum Bernoulli noises,
which can be viewed as a continuous-time quantum walk on a special locally infinite graph.
Our main work is as follows:
\begin{itemize}
  \item We introduce a model of continuous-time quantum walk based on quantum Bernoulli noises, which may describe the motion
    of a free quantum particle on a special locally infinite graph.
  \item We examine the spectral property of the model and obtain several spectral results accordingly.
  \item We investigate the time-reversal symmetry of the model. To our surprise, we find that the time-reversal symmetry of the model
       can be described directly by a unitary operator. This contrasts sharply with that in the classical theory of time-reversal symmetry.
  \item We show that the probability distributions of the model also have some kind of time-reversal symmetry.
\end{itemize}
Our work might suggest that, in some special cases, unitary operators can play the same role as anti-unitary operators
in dealing with time-reversal symmetry in quantum systems.

The paper is organized as follows. In Section~\ref{sec-2}, we prove some auxiliary results
and then define our model of continuous-time quantum walk based on quantum Bernoulli noises.
Section~\ref{sec-3} mainly examines the spectral property of the model and several spectral results are obtained therein.
Finally in Section~\ref{sec-4}, we investigate the time-reversal symmetry of the model.
We first construct a unitary operator and then show that this unitary operator can serve as a time-reversal operator
of the model. As application, we show that the probability distributions of the model also have some kind of time-reversal symmetry.

\textbf{Conventions and frequently used notation.}
Throughout this paper, $\mathbb{R}$ denotes the set of real numbers,
while $\mathbb{N}$ signifies the set of nonnegative integers, namely $\mathbb{N}:=\{0,\, 1,\, 2, \cdots\}$.
By a Hilbert space we mean a separable complex Hilbert space whose inner product is linear
in the second variable. Unless otherwise stated, an operator always means a linear operator.
Given a bounded operator $S$ on a Hilbert space $\mathcal{H}$, $S^*$ stands for its adjoint. $S$ is said to be positive, written $S\geq 0$,
if $\langle u,Su\rangle\geq 0$ for all $u\in \mathcal{H}$, where $\langle\cdot,\cdot\rangle$ denotes the inner product in $\mathcal{H}$.

\section{Model of continuous-time quantum walk}\label{sec-2}

This section introduces a model of continuous-time quantum walk based on quantum Bernoulli noises.
For more details about quantum Bernoulli noises, we refer to \cite{WTR-JMP-2019, W-JMP-2022} and references therein.

Let $\mathfrak{h}$ be the space of square integrable Bernoulli functionals, whose inner product and norm
are written $\langle\cdot,\cdot\rangle$ and $\|\cdot\|$, respectively. As is known, $\mathfrak{h}$ is of infinite dimension
and has an orthonormal basis of the form $\{Z_{\sigma} \mid \sigma\in \Gamma\}$, where $\Gamma$ signifies the finite power set of
$\mathbb{N}$, namely
\begin{equation}\label{eq-1-1}
  \Gamma = \big\{\sigma \mid \sigma\subset \mathbb{N},\, \#(\sigma)<\infty\big\},
\end{equation}
where $\#(\sigma)$ means the cardinality of $\sigma$ as a set. We call $\{Z_{\sigma} \mid \sigma\in \Gamma\}$
the canonical ONB for $\mathfrak{h}$. For convenience, we use $\sigma\setminus k$ to mean $\sigma\setminus \{k\}$
when $k\in \mathbb{N}$ and $\sigma\in \Gamma$. Similarly, we use $\sigma\cup k$ and $\sigma\triangle k$.

According to \cite{WCL-JMP-2010}, for each $k\in \mathbb{N}$, there exists a bounded operator $\partial_k$ on $\mathfrak{h}$
such that
\begin{equation}\label{eq-annihilation-creation-operator}
  \partial_kZ_{\sigma} = \mathbf{1}_{\sigma}(k) Z_{\sigma\setminus k},\quad
  \partial_k^*Z_{\sigma} = (1-\mathbf{1}_{\sigma}(k)) Z_{\sigma\cup k},\quad \sigma\in \Gamma,
\end{equation}
where $\mathbf{1}_{\sigma}(k)$ denotes the indicator of $\sigma$ as a subset of $\mathbb{N}$.
The family $\{\partial_k, \partial_k^* \mid k\in \mathbb{N}\}$ is then known as quantum Bernoulli noises (QBN),
and it satisfies the following commutation relations:
\begin{equation*}
  \partial_j\partial_k = \partial_k\partial_j,\quad \partial_j^*\partial_k^*=\partial_k^*\partial_j^*,\quad
  \partial_j\partial_k^* = \partial_k^*\partial_j\quad (j,\, k\in \mathbb{N},\, j\ne k)
\end{equation*}
and
\begin{equation}\label{eq}
 \partial_k\partial_k=\partial_k^*\partial_k^*=0,\quad \partial_k\partial_k^*+ \partial_k^*\partial_k=I\quad (k\in \mathbb{N}),
\end{equation}
where $I$ signifies the identity operator on $\mathfrak{h}$. In particular, it satisfies the canonical anti-commutation relations (CAR)
in equal time.

In what follows, we write $\Xi_k = \partial_k^* + \partial_k$ for $k\in \mathbb{N}$. It then follows from the above commutation relations that
each $\Xi_k$ is a self-adjoint unitary operator and the family $\{\Xi_k \mid k\in \mathbb{N}\}$ is commuting in the sense that
$\Xi_j\Xi_k=\Xi_k\Xi_j$ for all $j$, $k\in \mathbb{N}$. Additionally, from (\ref{eq-annihilation-creation-operator}),
one has
\begin{equation}\label{eq-shift-property}
  \Xi_kZ_{\sigma} = Z_{\sigma \triangle k},\quad k\in \mathbb{N},\, \sigma\in \Gamma,
\end{equation}
where $Z_{\sigma}$ and $Z_{\sigma \triangle k}$ are the basis vectors of the canonical ONB for $\mathfrak{h}$, respectively.

\begin{definition}
A function $w\colon \mathbb{N} \rightarrow(0,\infty)$ is called a weight on $\mathbb{N}$
if $|w|:=\sum_{k=0}^{\infty}w(k)<\infty$. For a weight $w$ on $\mathbb{N}$, we define
\begin{equation}\label{eq-adjacency-operator}
  A_w=\sum_{k=0}^{\infty}w(k)\Xi_k
\end{equation}
and call it the $w$-adjacency operator.
\end{definition}

Note that the operator series in (\ref{eq-adjacency-operator}) is convergent in operator norm,
hence the $w$-adjacency operator $A_w$ is a bounded operator on $\mathfrak{h}$.

\begin{theorem}\label{thr-basic-2-1}
Let $w$ be a weight on $\mathbb{N}$. Then $A_w$ is self-adjoint and moreover $\|A_w\| = |w|$.
\end{theorem}

\begin{proof}
It is easy to see that $A_w$ is self-adjoint. Next let us show that $\|A_w\| = |w|$.
Let $n\geq 1$ be a positive integer. Consider the vector $\xi_n\in \mathfrak{h}$ given by
\begin{equation}\label{eq}
  \xi_n = \sum_{\tau\in \Gamma_n}Z_{\tau},
\end{equation}
where $\Gamma_n:=\{\tau \mid \tau\subset \mathbb{N}_n\}$ with $\mathbb{N}_n:=\{0, 1, \cdots, n\}$.
For each $k\in \mathbb{N}_n$, since the mapping $\tau\mapsto \tau\triangle k$ is a bijection from $\Gamma_n$ to itself,
using (\ref{eq-shift-property}) we have
\begin{equation*}
  \Xi_k\xi_n = \sum_{\tau\in \Gamma_n}\Xi_kZ_{\tau}
             = \sum_{\tau\in \Gamma_n}Z_{\tau \triangle k}
             = \sum_{\tau\in \Gamma_n}Z_{\tau}
             = \xi_n,
\end{equation*}
which implies that
\begin{equation*}
  \Big(\sum_{k=0}^nw(k)\Xi_k\Big)\xi_n = \Big(\sum_{k=0}^nw(k)\Big)\xi_n.
\end{equation*}
Thus, we further have
\begin{equation*}
  \Big(\sum_{k=0}^nw(k)\Big)\|\xi_n\|
  =\Big\|\Big(\sum_{k=0}^nw(k)\Big)\xi_n\Big\|
  = \Big\|\Big(\sum_{k=0}^nw(k)\Xi_k\Big)\xi_n\Big\|
  \leq \Big\|\sum_{k=0}^nw(k)\Xi_k\Big\|\|\xi_n\|,
\end{equation*}
which, together with $\|\xi_n\|> 0$, gives
\begin{equation*}
 \sum_{k=0}^nw(k) \leq \Big\|\sum_{k=0}^nw(k)\Xi_k\Big\|.
\end{equation*}
It then follows from taking the limit (as $n\to \infty$) that $|w| \leq \|A_w\|$. Obviously, $\|A_w\|\leq |w|$.
Therefore, we finally have $\|A_w\|= |w|$.
\end{proof}

As shown above, given a weight $w$ on $\mathbb{N}$, one has a self-adjoint bounded operator $A_w$ on $\mathfrak{h}$.
To help get an insight into its physical meaning, let us give a graph-theoretic interpretation to $A_w$ bellow.

Consider the set $\Gamma$, which is defined in (\ref{eq-1-1}).
We introduce an adjacency relation in $\Gamma$ as follows: two elements $\sigma$, $\tau\in \Gamma$ are said to be adjacent if $\#(\sigma\triangle\tau)=1$, where $\sigma\triangle\tau$ signifies the symmetric difference of $\sigma$ and $\tau$ as subsets of $\mathbb{N}$.
In that case, we write $\sigma\sim \tau$. With this adjacency relation, $\Gamma$ forms an infinite graph,
which we denote by $(\Gamma,\sim)$.

\begin{lemma}\cite{WTR-JMP-2019}\label{lem-2-2}
Let $n$ be a positive integer and $\Gamma_n=\{\sigma \mid \sigma\subset \mathbb{N}_n\}$ with $N_n=\{0, 1,\cdots, n\}$.
Then $\Gamma_n\subset \Gamma$. Moreover, as a subgraph of the graph $(\Gamma,\sim)$, $(\Gamma_n,\sim)$ is isomorphic to
the $n+1$-dimensional hypercube.
\end{lemma}

It is easy to see that $\Gamma_1\subset \Gamma_2\subset \cdots \subset\Gamma_n \subset \Gamma_{n+1}\subset \cdots \subset\Gamma$.
Additionally, one can also show that
\begin{equation}
  \Gamma = \bigcup_{n=1}^{\infty}\Gamma_n.
\end{equation}
This, together with Lemma~\ref{lem-2-2}, suggests that the graph $(\Gamma,\sim)$ may be called a hypercube of infinite-dimension.

\begin{proposition}
Let $\sigma$, $\tau\in \Gamma$ be vertices of the graph $(\Gamma,\sim)$. Then, $\sigma\sim \tau$ if and only if there exists
a unique $k\in \mathbb{N}$ such that $\Xi_kZ_{\sigma} = Z_{\tau}$.
\end{proposition}

\begin{proof}
The ``only if'' part. Suppose that $\sigma\sim \tau$. Then, by definition, $\#(\sigma\triangle \tau) =1$, which implies that
there exists a unique $k\in \mathbb{N}$ such that $\sigma\triangle \tau=\{k\}$, which is equivalent to
$\tau=\sigma\triangle k$. Hence, by (\ref{eq-shift-property}), we come to $\Xi_kZ_{\sigma}=Z_{\sigma\triangle k}=Z_{\tau}$.

The ``if'' part. Suppose that $\Xi_kZ_{\sigma} = Z_{\tau}$. Then, again using (\ref{eq-shift-property}), we find
$Z_{\sigma\triangle k} = Z_{\tau}$, which implies $\sigma\triangle k= \tau$, which is equivalent to $\sigma\triangle \tau= \{k\}$.
Thus, $\#(\sigma\triangle\tau)=1$, namely $\sigma\sim \tau$.
\end{proof}

\begin{corollary}\label{coro-neighbour}
Let $\sigma\in \Gamma$ be a vertex of the graph $(\Gamma,\sim)$ and write $\mathcal{N}(\sigma)=\{\tau\in \Gamma \mid \tau\sim \sigma\}$. Then
$\mathcal{N}(\sigma)$ has a representation of the following form
\begin{equation}\label{eq}
  \mathcal{N}(\sigma) = \{\,\sigma\triangle k \mid k\in \mathbb{N}\,\}.
\end{equation}
In particular, the graph $(\Gamma,\sim)$ is locally infinite.
\end{corollary}

Let $l^2(\Gamma)$ be the space of square summable complex-valued functions defined on $\Gamma$.
Then, as a Hilbert space, $l^2(\Gamma)$ has an orthonormal basis of the form $\{\psi_{\sigma} \mid \sigma\in \Gamma\}$,
where $\psi_{\sigma}$ is the function on $\Gamma$ given by
\begin{equation}\label{eq}
  \psi_{\sigma}(\gamma) =
  \left\{
    \begin{array}{ll}
      1, & \hbox{$\gamma=\sigma$;} \\
      0, & \hbox{$\gamma\in \Gamma$, $\gamma\ne \sigma$.}
    \end{array}
  \right.
\end{equation}
Clearly, there exists a unitary isomorphism $\mathsf{F}\colon l^2(\Gamma)\rightarrow \mathfrak{h}$ such that
\begin{equation}\label{eq}
  \mathsf{F}\psi_{\sigma} = Z_{\sigma},\quad \sigma\in \Gamma.
\end{equation}

\begin{theorem}\label{thr-adjacency-represent}
Let $w$ be a weight on $\mathbb{N}$. Then, as a self-adjoint bounded operator on space $l^2(\Gamma)$,
$\widetilde{A}_w:=\mathsf{F}^{-1}A_w\mathsf{F}$ has a representation of the form
\begin{equation}\label{eq-adjacency-represent}
  \widetilde{A}_wf(\sigma) = \sum_{k=0}^{\infty}w(k)f(\sigma\triangle k),\quad \sigma\in \Gamma,
\end{equation}
where $f\in l^2(\Gamma)$.
\end{theorem}

\begin{equation*}
\begin{diagram}
\node{l^2(\Gamma)}\arrow[2]{e,t}{\mathsf{F}}\arrow{s,l}{\widetilde{A}_w\,=\,\mathsf{F}^{-1}A_w\mathsf{F}} \node[2]{\mathfrak{h}}\arrow{s,r}{A_w}\\
\node{l^2(\Gamma)}\node[2]{\mathfrak{h}}\arrow[2]{w,b}{\mathsf{F}^{-1}}
\end{diagram}
\end{equation*}

\begin{proof}
For $k\in \mathbb{N}$, by setting $\widetilde{\Xi}_k=\mathsf{F}^{-1}\Xi_k\mathsf{F}$, we get a self-adjoint unitary operator
$\widetilde{\Xi}_k$ on $l^2(\Gamma)$. It follows from (\ref{eq-adjacency-operator}) that
\begin{equation}\label{eq-adjacency-operator-2}
  \widetilde{A}_w
  = \mathsf{F}^{-1}A_w\mathsf{F}
  = \sum_{k=0}^{\infty}w(k)\widetilde{\Xi}_k,
\end{equation}
where the series converges in operator norm. Hence $\widetilde{A}_w$ is a self-adjoint bounded operator on $l^2(\Gamma)$.
Now let $f\in l^2(\Gamma)$ and $\sigma\in \Gamma$ be given.
For each $k\in \mathbb{N}$, using the Fourier expansion of $f$ with respect to the orthonormal basis $\{\psi_{\sigma} \mid \sigma\in \Gamma\}$,
we have
\begin{equation*}
  \widetilde{\Xi}_kf
 = \widetilde{\Xi}_k\Big(\sum_{\sigma\in \Gamma}f(\sigma)\psi_{\sigma}\Big)
 = \sum_{\sigma\in \Gamma}f(\sigma)\widetilde{\Xi}_k\psi_{\sigma},
\end{equation*}
which, together with $\widetilde{\Xi}_k\psi_{\sigma}= \psi_{\sigma \triangle k}$, gives
\begin{equation*}
  \widetilde{\Xi}_kf
  = \sum_{\sigma\in \Gamma}f(\sigma)\psi_{\sigma \triangle k}
  = \sum_{\tau\in \Gamma}f(\tau \triangle k)\psi_{\tau}.
\end{equation*}
Thus, for each $k\in \mathbb{N}$, we have $\widetilde{\Xi}_kf(\sigma) = f(\sigma\triangle k)$.
This together with (\ref{eq-adjacency-operator-2}) implies that
\begin{equation*}
  \widetilde{A}_wf(\sigma)
  = \sum_{k=0}^{\infty}w(k)\widetilde{\Xi}_kf(\sigma)
  = \sum_{k=0}^{\infty}w(k)f(\sigma\triangle k),
\end{equation*}
which is the desired.
\end{proof}

\begin{remark}\label{rem-2-1}
Consider the subgraph $(\Gamma_n,\sim)$ of the graph $(\Gamma,\sim)$, where $n\geq 1$. As indicated in Lemma~\ref{lem-2-2},
$(\Gamma_n,\sim)$ is isomorphic to the $n+1$-dimensional hypercube. Let $\widetilde{A}$ be the adjacency operator (matrix) of $(\Gamma_n,\sim)$.
Then, by the spectral theory of graphs \cite{obata} and the structure of $(\Gamma_n,\sim)$, one can show that
$\widetilde{A}$ has a representation of the form
\begin{equation*}
  \widetilde{A}f(\sigma) = \sum_{k=0}^nf(\sigma\triangle k),\quad \sigma\in \Gamma_n,
\end{equation*}
where $f\in l^2(\Gamma_n)$.
Comparing this with (\ref{eq-adjacency-represent}), we come to the observation that
$\widetilde{A}_w=\mathsf{F}^{-1}A_w\mathsf{F}$ is actually an adjacency operator (matrix) of the graph $(\Gamma,\sim)$.
In other words, $A_w$ can be thought of as an adjacency operator (matrix) of the graph $(\Gamma,\sim)$.
This justifies the name of $A_w$ and the next definition.
\end{remark}

\begin{definition}
Let $w$ be a weight on $\mathbb{N}$. The continuous-time quantum walk $A_w$ (the walk $A_w$ below)
is a continuous-time quantum walk that admits the following features:
\begin{itemize}
  \item[(1)]\  The state space of the walk is $\mathfrak{h}$ and its states are represented by unit vectors in $\mathfrak{h}$;
  \item[(2)]\  The evolution of the walk is governed by equation
    \begin{equation}\label{eq-qw-equation}
        \xi_t = \mathrm{e}^{-\mathrm{i}tA_w}\xi_0,\quad t\in \mathbb{R},
    \end{equation}
      where $\xi_t$ denotes its state at time $t$, especially $\xi_0$ is its initial state;
  \item[(3)]\ The probability $P_t(\sigma\!\mid\!\xi_0)$ that the walker is found at vertex $\sigma\in \Gamma$ at time $t\in \mathbb{R}$ is given by
   \begin{equation}\label{eq-3-11}
   P_t(\sigma\!\mid\!\xi_0) = |\langle Z_{\sigma}, \xi_t\rangle|^2.
   \end{equation}
\end{itemize}
In that case, the collection $\{\,\xi_t \mid t\in \mathbb{R}\,\}$ is called the trajectory of the walk with initial state $\xi_0$,
while, for $t\in \mathbb{R}$, the function $\sigma\mapsto P_t(\sigma\!\mid\!\xi_0)$ on $\Gamma$ is referred to as the probability distribution
of the walk at time $t$.
\end{definition}

In light of Remark~\ref{rem-2-1}, the walk $A_w$ can be naturally viewed as a model of continuous-time quantum walk on the graph $(\Gamma,\sim)$.
Physically, it may describe the motion of a free quantum particle on the graph $(\Gamma,\sim)$.

\section{Spectral property}\label{sec-3}

Let $w$ be a weight on $\mathbb{N}$. In this section, we mainly examine the spectral property of the walk $A_w$.
In the following, we use $\mathrm{spec}\,(S)$ to mean the spectrum of an operator $S$.

Let $2^{\mathbb{N}}$ be the power set of $\mathbb{N}$, namely $2^{\mathbb{N}}:= \{\,\sigma \mid \sigma\subset \mathbb{N}\,\}$.
With the weight $w$, we can associate a function $\mu_w$ on $2^{\mathbb{N}}$ in the following manner:
\begin{equation}\label{eq-3-1}
  \mu_w(\sigma) = \sum_{k=0}^{\infty}\mathbf{1}_{\sigma}(k)w(k),\quad \sigma\in 2^{\mathbb{N}}.
\end{equation}
Note that $\mu_w(\sigma)$ still makes sense for $\sigma \in \Gamma$ since $\Gamma\subset 2^{\mathbb{N}}$.

Clearly, the function $\mu_w$ is bounded by the interval $[0,|w|]$ in the sense that
$0\leq \mu_w(\sigma)\leq |w|$, $\sigma\in 2^{\mathbb{N}}$. The weight $w$ is said to be ideal if it satisfies that
\begin{equation}\label{eq}
\overline{\big\{\mu_w(\sigma)\mid \sigma\in 2^{\mathbb{N}}\big\}}=[0,|w|],
\end{equation}
where $\overline{\big\{\mu_w(\sigma)\mid \sigma\in 2^{\mathbb{N}}\big\}}$ means the closure of
$\big\{\mu_w(\sigma)\mid \sigma\in 2^{\mathbb{N}}\big\}$ in $\mathbb{R}$.

\begin{proposition}\label{prop-shift-eigenvector}
Let $\sigma\in \Gamma$ be given. Then, correspondingly, there exists a sequence $(u_n)_{n\geq 1}$ of unit vectors in $\mathfrak{h}$ such that
\begin{equation}\label{eq-shift-eugenvector}
  \Xi_ku_n = \mathcal{E}_{\sigma}(k)u_n,\quad k\in \mathbb{N}_n,\, n\geq 1,
\end{equation}
where $\mathcal{E}_{\sigma}(k)=2\times\mathbf{1}_{\sigma}(k)-1$ and $\mathbb{N}_n=\{0, 1,\cdots,n\}$.
\end{proposition}

\begin{proof}
For each $n\geq 1$, define $u_n \in \mathfrak{h}$ as
\begin{equation}\label{eq-unit-vector}
  u_n = \frac{1}{\sqrt{2^{n+1}}}\Big(\prod_{k=0}^n\big(I + \mathcal{E}_{\sigma}(k)\Xi_k\big)\Big)Z_{\emptyset},
\end{equation}
where $Z_{\emptyset}$ is the basis vector of the canonical ONB for $\mathfrak{h}$ that is indexed by $\emptyset$.
Then, we get a sequence $(u_n)_{n\geq 1}$ in $\mathfrak{h}$. Next, let us show that the sequence is the desired.
Let $n\geq 1$ be given. A lengthy calculation gives
\begin{equation*}
  u_n = \frac{1}{\sqrt{2^{n+1}}}\sum_{\tau\in \Gamma_n}\Big(\prod_{k\in \tau}\mathcal{E}_{\sigma}(k)\Big)Z_{\tau},
\end{equation*}
where $\Gamma_n = \{\tau \mid \tau \subset \mathbb{N}_n\}$. Thus
\begin{equation*}
  \|u_n\|^2 = \frac{1}{2^{n+1}}\sum_{\tau\in \Gamma_n}\Big|\prod_{k\in \tau}\mathcal{E}_{\sigma}(k)\Big|^2\|Z_{\tau}\|^2
            = \frac{1}{2^{n+1}}\sum_{\tau\in \Gamma_n}1
            = 1,
\end{equation*}
which shows that $u_n$ is a unit vector. Now let $k\in \mathbb{N}_n$. By the commuting property of
the operator family $\{\Xi_j \mid j\in \mathbb{N}\}$ and the equality
$\Xi_k(I+\mathcal{E}_{\sigma}(k)\Xi_k) = \mathcal{E}_{\sigma}(k)(I+\mathcal{E}_{\sigma}(k)\Xi_k)$, we have
\begin{equation*}
\begin{split}
  \Xi_k\Big(\prod_{j=0}^n\big(I + \mathcal{E}_{\sigma}(j)\Xi_j\big)\Big)
    & = \Xi_k\big(I + \mathcal{E}_{\sigma}(k)\Xi_k\big)\Big(\prod_{j=0, j\ne k}^n\big(I + \mathcal{E}_{\sigma}(j)\Xi_j\big)\Big)\\
    & = \mathcal{E}_{\sigma}(k)(I+\mathcal{E}_{\sigma}(k)\Xi_k)\Big(\prod_{j=0, j\ne k}^n\big(I + \mathcal{E}_{\sigma}(j)\Xi_j\big)\Big)\\
    & = \mathcal{E}_{\sigma}(k)\Big(\prod_{j=0}^n\big(I + \mathcal{E}_{\sigma}(j)\Xi_j\big)\Big),
\end{split}
\end{equation*}
which together with (\ref{eq-unit-vector}) yields
\begin{equation*}
  \Xi_ku_n
  = \frac{1}{\sqrt{2^{n+1}}}\Xi_k\Big(\prod_{k=0}^n\big(I + \mathcal{E}_{\sigma}(k)\Xi_k\big)\Big)Z_{\emptyset}
  = \frac{\mathcal{E}_{\sigma}(k)}{\sqrt{2^{n+1}}}\Big(\prod_{k=0}^n\big(I + \mathcal{E}_{\sigma}(k)\Xi_k\big)\Big)Z_{\emptyset}
  = \mathcal{E}_{\sigma}(k)u_n.
\end{equation*}
This completes the proof.
\end{proof}

\begin{theorem}\label{thr-basic-3-2}
Write $B_w = \frac{1}{2}(|w|I + A_w)$. Then $B_w$ is a positive bounded operator and its spectrum admits an estimate of the following form
\begin{equation}\label{eq-spectrum-estimate}
 \overline{\{\mu_w(\sigma) \mid \sigma\in 2^{\mathbb{N}}\}} \subset  \mathrm{spec}\,(B_w)\subset [0,|w|].
\end{equation}
In particular, if the weight $w$ is ideal, then $\mathrm{spec}\,(B_w) = [0,|w|]$.
\end{theorem}

\begin{proof}
Clearly, $B_w$ is bounded. For each $k\in \mathbb{N}$, it can be verified that $\frac{1}{2}(I+\Xi_k)$ is self-adjoint and idempotent,
hence it is a projection operator, which implies that $I+\Xi_k\geq 0$.
On the other hand, by straightforward calculations, we find
\begin{equation}\label{eq-3-6}
  B_w = \frac{1}{2}\sum_{k=0}^{\infty}w(k)(I + \Xi_k),
\end{equation}
which, together with $w(k)>0$ as well as $I+\Xi_k\geq 0$, means that $B_w\geq 0$, namely $B_w$ is positive.
Using Theorem~\ref{thr-basic-2-1}, we get
\begin{equation*}
\|B_w\|= \frac{1}{2}\|(|w|I + A_w)\|\leq \frac{1}{2}(|w| + \|A_w\|) = |w|.
\end{equation*}
Thus, by the spectral theory of operators \cite{reed-simon}, we come to know that $\mathrm{spec}\,(B_w)\subset [0,|w|]$.

Let $\sigma\in \Gamma$ be given. By Proposition~\ref{prop-shift-eigenvector}, there exists a sequence $(u_n)_{n\geq 1}$ of unit vectors
in $\mathfrak{h}$ such that (\ref{eq-shift-eugenvector}) holds.
For each $n\geq 1$, by (\ref{eq-3-6}) and Proposition~\ref{prop-shift-eigenvector}, we have
\begin{equation*}
  B_wu_n = \Big(\sum_{k=0}^n\mathbf{1}_{\sigma}(k)w(k)\Big)u_n + \frac{1}{2}\sum_{k=n+1}^{\infty}w(k)(I+\Xi_k)u_n.
\end{equation*}
On the other hand, for each $n\geq 1$, by (\ref{eq-3-1}) we have
\begin{equation*}
  \mu_w(\sigma)u_n = \Big(\sum_{k=0}^n\mathbf{1}_{\sigma}(k)w(k)\Big)u_n + \Big(\sum_{k=n+1}^{\infty}\mathbf{1}_{\sigma}(k)w(k)\Big)u_n.
\end{equation*}
Thus, for each $n\geq 1$, using the above two equalities we get
\begin{equation*}
\begin{split}
  \|B_wu_n- \mu_w(\sigma)u_n\|
   & = \Big\| \frac{1}{2}\sum_{k=n+1}^{\infty}w(k)(I+\Xi_k)u_n
          - \Big(\sum_{k=n+1}^{\infty}\mathbf{1}_{\sigma}(k)w(k)\Big)u_n\Big\|\\
   & \leq \frac{1}{2}\sum_{k=n+1}^{\infty}w(k)\|I+\Xi_k\|\|u_n\|
          + \Big(\sum_{k=n+1}^{\infty}\mathbf{1}_{\sigma}(k)w(k)\Big)\|u_n\|\\
   & \leq 2 \sum_{k=n+1}^{\infty}w(k).
\end{split}
\end{equation*}
This, together with the property $\sum_{k=0}^{\infty}w(k)<\infty$, implies that $\|B_wu_n- \mu_w(\sigma)u_n\|\rightarrow 0$ as $n\to \infty$.
Thus, by the well-known Weyl's criterion (see Theorem VII.12 of \cite{reed-simon}),
we come to know that $\mu_w(\sigma)\in \mathrm{spec}\,(B_w)$. It then follows from the arbitrariness of $\sigma\in \Gamma$ that
\begin{equation*}
\{\mu_w(\sigma) \mid \sigma\in \Gamma\} \subset  \mathrm{spec}\,(B_w),
\end{equation*}
which, together with the inclusion relation
$\big\{\mu_w(\sigma)\mid \sigma\in 2^{\mathbb{N}}\big\}\subset \overline{\{\mu_w(\sigma) \mid \sigma\in \Gamma\}}$,
yields the desired inclusion relation
$\overline{\big\{\mu_w(\sigma)\mid \sigma\in 2^{\mathbb{N}}\big\}}\subset \mathrm{spec}\,(B_w)$.

Finally, if the weight $w$ is ideal, then $\overline{\{\mu_w(\sigma)\mid \sigma\in 2^{\mathbb{N}}\}}=[0,|w|]$,
which together with (\ref{eq-spectrum-estimate}) implies $\mathrm{spec}\,(B_w) = [0,|w|]$.
\end{proof}

\begin{corollary}\label{coro-3-3}
Let the weight $w$ be ideal. Then it holds true that
\begin{equation}\label{eq}
\mathrm{spec}\,(B_w) = \overline{\big\{\mu_w(\sigma) \mid \sigma\in 2^{\mathbb{N}}\big\}}= \overline{\{\mu_w(\sigma) \mid \sigma\in \Gamma\}}.
\end{equation}
\end{corollary}

We now turn to examining the spectral property of the walk $A_w$, which is closely related to that of the operator $B_w$ as will be seen below.

\begin{theorem}\label{thr-basic-3-3}
The $w$-adjacency operator $A_w$ has a spectral estimate of the following form
\begin{equation}\label{eq-basic-3-3}
  \overline{\big\{2\mu_w(\sigma)-|w| \mid \sigma\in 2^{\mathbb{N}}\big\}}
  \subset \mathrm{spec}\,(A_w) \subset [-|w|,|w|].
\end{equation}
In particular, if the weight $w$ is ideal, then $\mathrm{spec}\,(A_w) = [-|w|, |w|]$.
\end{theorem}

\begin{proof}
Consider the function $f(t) = 2t- |w|$. It is easy to see that $A_w = f(B_w)$.
Hence, by the spectral mapping theorem \cite{reed-simon} and Theorem~\ref{thr-basic-3-2} above, we get (\ref{eq-basic-3-3}).
If the weight $w$ is ideal, then again by the spectral mapping theorem and Theorem~\ref{thr-basic-3-2},
we get $\mathrm{spec}\,(A_w) = [-|w|, |w|]$.
\end{proof}

\begin{remark}
Consider the weight $w_0$ on $\mathbb{N}$ given by $w_0(k) = \frac{1}{2^{k+1}}$, $k\in \mathbb{N}$. It is not hard to verify
that $w_0$ is ideal. Thus, the conditions required in Theorem~\ref{thr-basic-3-2}, Corollary~\ref{coro-3-3} and Theorem~\ref{thr-basic-3-3}
can be satisfied.
\end{remark}

For $\sigma\in \Gamma$ with $\sigma\ne \emptyset$, we define $\Xi_{\sigma}=\prod_{k\in \sigma} \Xi_k$, which makes sense since
the family $\{\Xi_k \mid k\in \mathbb{N}\}$ is commuting. By convention, we set $\Xi_{\emptyset}=I$.
Using (\ref{eq-shift-property}), we can show that
\begin{equation}\label{eq}
  Z_{\sigma} = \Xi_{\sigma}Z_{\emptyset},\quad \sigma\in \Gamma.
\end{equation}
Note that $\Xi_{\sigma}$ is self-adjoint for all $\sigma\in \Gamma$.
Keeping these facts in mind, we can prove the next theorem, which provides some information about the spectral structure of $A_w$.

\begin{theorem}\label{thr-basic-3-4}
Both $-|w|$ and $|w|$ belong to $\mathrm{spec}\,(A_w)$, but they are not eigenvalues of $A_w$.
\end{theorem}

\begin{proof}
It is easy to see that $-|w| = 2\mu_w(\emptyset)-|w|$ and $|w| = 2\mu_w(\mathbb{N})-|w|$, hence by (\ref{eq-basic-3-3})
we know that $-|w|$ and $|w|$ belongs to $\mathrm{spec}\,(A_w)$. Next, we show that $|w|$ is not an eigenvalue of $A_w$.

Let $\xi\in \mathfrak{h}$ be such that $A_w\xi = |w|\xi$. Then, using $(I-\Xi_k)^2= 2(I-\Xi_k)$, $k\in \mathbb{N}$,
we have
\begin{equation*}
  \begin{split}
   \frac{1}{2}\sum_{k=0}^{\infty}w(k)\|(I-\Xi_k)\xi\|^2
    & = \frac{1}{2}\sum_{k=0}^{\infty}w(k)\langle \xi, (I-\Xi_k)^2\xi\rangle\\
    & = \sum_{k=0}^{\infty}w(k)\langle \xi, (I-\Xi_k)\xi\rangle\\
    & = \langle\xi, |w|\xi - A_w\xi\rangle\\
    & =0,
  \end{split}
\end{equation*}
which, together with $w(k)>0$, $k\in \mathbb{N}$, implies that $\Xi_k\xi=\xi$, $\forall\, k\in \mathbb{N}$.
Thus, $\Xi_{\sigma}\xi =\xi$, $\forall\, \sigma\in \Gamma$, which together with Parseval formula gives
\begin{equation*}
  \|\xi\|^2
   = \sum_{\sigma\in \Gamma}|\langle Z_{\sigma},\xi\rangle|^2
   = \sum_{\sigma\in \Gamma}|\langle \Xi_{\sigma}Z_{\emptyset},\xi\rangle|^2
   = \sum_{\sigma\in \Gamma}|\langle Z_{\emptyset},\Xi_{\sigma}\xi\rangle|^2
   = \sum_{\sigma\in \Gamma}|\langle Z_{\emptyset},\xi\rangle|^2,
\end{equation*}
which implies that $\xi=0$. Thus, $|w|$ is not an eigenvalue of $A_w$.

By using the operator $T$ introduced in Section~\ref{sec-4}
and Proposition~\ref{prop-shift-time-commutation} therein, we can show that $-|w|$ is not an eigenvalue of $A_w$ either.
\end{proof}

\section{Time-reversal symmetry}\label{sec-4}

Let $w$ be a fixed weight on $\mathbb{N}$. In the present section, we investigate the time-reversal symmetry of the walk $A_w$
and properties of its probability distributions.

\begin{theorem}\label{thr-time-reversal-operator}
There exists a unitary operator $T$ on $\mathfrak{h}$ such that $T^*=T$ and
\begin{equation}\label{eq-time-reversal-operator}
  TZ_{\sigma} = (-1)^{\#(\sigma)}Z_{\sigma},\quad \sigma\in \Gamma.
\end{equation}
\end{theorem}

\begin{proof}
For any $\xi\in \mathfrak{h}$, since the family $\{Z_{\sigma}\mid \sigma\in \Gamma\}$ is the canonical ONB for $\mathfrak{h}$,
we find
\begin{equation*}
  \sum_{\sigma\in \Gamma}\big|(-1)^{\#(\sigma)}\langle Z_{\sigma},\xi\rangle\big|^2
= \sum_{\sigma\in \Gamma}\big|\langle Z_{\sigma},\xi\rangle\big|^2
= \|\xi\|^2 <\infty,
\end{equation*}
which implies that the vector series $\sum_{\sigma\in \Gamma}(-1)^{\#(\sigma)}\langle Z_{\sigma},\xi\rangle Z_{\sigma}$
converges in the norm of $\mathfrak{h}$. In view of this, we define a mapping $T\colon \mathfrak{h}\rightarrow \mathfrak{h}$ as follows:
\begin{equation*}
T\xi = \sum_{\sigma\in \Gamma}(-1)^{\#(\sigma)}\langle Z_{\sigma},\xi\rangle Z_{\sigma},\quad \xi\in \mathfrak{h}.
\end{equation*}
Clearly, $T$ is a bounded operator on $\mathfrak{h}$ and satisfies (\ref{eq-time-reversal-operator}).
For any $\xi$, $\eta\in \mathfrak{h}$, straightforward calculations give
\begin{equation*}
  \langle T\xi, \eta\rangle
  = \sum_{\sigma\in \Gamma}\overline{(-1)^{\#(\sigma)}\langle Z_{\sigma},\xi\rangle} \langle Z_{\sigma},\eta\rangle
  = \sum_{\sigma\in \Gamma}\overline{\langle Z_{\sigma},\xi\rangle} (-1)^{\#(\sigma)}\langle Z_{\sigma},\eta\rangle
  = \langle \xi, T\eta\rangle,
\end{equation*}
which means that $T$ is symmetric, hence $T^*=T$. Finally, for any $\xi\in \mathfrak{h}$, using (\ref{eq-time-reversal-operator}) we have
\begin{equation*}
  T^2\xi
   = T(T\xi)
   = \sum_{\sigma\in \Gamma}(-1)^{\#(\sigma)}\langle Z_{\sigma},T\xi\rangle Z_{\sigma}
   = \sum_{\sigma\in \Gamma}(-1)^{\#(\sigma)}\langle TZ_{\sigma},\xi\rangle Z_{\sigma}
   = \sum_{\sigma\in \Gamma}\langle Z_{\sigma},\xi\rangle Z_{\sigma}
   =\xi.
\end{equation*}
Thus $T^2=I$, which together with $T^*=T$ implies that $T$ is a unitary operator.
\end{proof}

By writing $2\mathbb{N}:= \{2k \mid k\in \mathbb{N}\}$ and $2\mathbb{N}+1:= \{2k+1 \mid k\in \mathbb{N}\}$,
we introduce two closed linear subspaces of $\mathfrak{h}$ as follows
\begin{equation}\label{eq-4-2}
  \mathfrak{h}^{(+)}= \overline{\mathrm{span}}\{\,Z_{\sigma} \mid \sigma\in \Gamma,\, \#(\sigma)\in 2\mathbb{N}\,\},\quad
  \mathfrak{h}^{(-)}= \overline{\mathrm{span}}\{\,Z_{\sigma} \mid \sigma\in \Gamma,\, \#(\sigma)\in 2\mathbb{N}+1\,\}.
\end{equation}
Here $\overline{\mathrm{span}}\, D$ signifies the closure of the linear subspace spanned by a vector set $D$ of $\mathfrak{h}$.
Clearly, $\mathfrak{h}^{(+)}$ and $\mathfrak{h}^{(-)}$ give rise to an orthogonal decomposition of $\mathfrak{h}$:
\begin{equation}
 \mathfrak{h}= \mathfrak{h}^{(+)} + \mathfrak{h}^{(-)},\quad \mathfrak{h}^{(+)}\perp \mathfrak{h}^{(-)}.
\end{equation}

\begin{proposition}\label{prop-4-2}
It holds that\ \ $\mathfrak{h}^{(+)}=\ker(I-T)$\ \ and\ \ $\mathfrak{h}^{(-)}=\ker(I+T)$.
\end{proposition}

\begin{proof}
Suppose that $\xi\in \mathfrak{h}^{(+)}$. Then, by (\ref{eq-4-2}) and Theorem~\ref{thr-time-reversal-operator}, we have
\begin{equation*}
  T\xi = \sum_{\sigma\in \Gamma, \#(\sigma)\in 2\mathbb{N}}\langle Z_{\sigma},\xi\rangle TZ_{\sigma}
       = \sum_{\sigma\in \Gamma, \#(\sigma)\in 2\mathbb{N}}\langle Z_{\sigma},\xi\rangle Z_{\sigma}
       = \xi,
\end{equation*}
which implies that $\xi\in \ker(I-T)$. Now assume that $\xi\in \ker(I-T)$. Then, for $\sigma\in \Gamma$ with $\#(\sigma)\in 2\mathbb{N}+1$,
we have
\begin{equation*}
\langle Z_{\sigma}, \xi\rangle = \langle Z_{\sigma}, T\xi\rangle = \langle T Z_{\sigma}, \xi\rangle
= \langle (-1)^{\#(\sigma)}Z_{\sigma}, \xi\rangle = -\langle Z_{\sigma}, \xi\rangle,
\end{equation*}
which implies that $\langle Z_{\sigma}, \xi\rangle =0$. Thus, using the Fourier expansion of $\xi$ and the above fact, we get
\begin{equation*}
  \xi = \sum_{\sigma\in \Gamma}\langle Z_{\sigma}, \xi\rangle Z_{\sigma}
      = \sum_{\sigma\in \Gamma,\#(\sigma)\in 2\mathbb{N}}\langle Z_{\sigma}, \xi\rangle Z_{\sigma},
\end{equation*}
which means that $\xi \in \mathfrak{h}^{(+)}$. Therefore $\mathfrak{h}^{(+)}=\ker(I-T)$. Similarly,
$\mathfrak{h}^{(-)}=\ker(I+T)$ can also be verified.
\end{proof}

\begin{proposition}\label{prop-shift-time-commutation}
For each $k\in \mathbb{N}$, it holds true that\ \ $T\Xi_k = -\Xi_kT$.
\end{proposition}

\begin{proof}
Let $k\in \mathbb{N}$ be given. Then, for each $\sigma\in \Gamma$, by (\ref{eq-shift-property}) and (\ref{eq-time-reversal-operator}) we have
\begin{equation*}
  T\Xi_kZ_{\sigma} = TZ_{\sigma \triangle k}=(-1)^{\#(\sigma \triangle k)}Z_{\sigma \triangle k}
= (-1)^{\#(\sigma \triangle k)}\Xi_kZ_{\sigma},
\end{equation*}
which, together with $(-1)^{\#(\sigma \triangle k)} = - (-1)^{\#(\sigma )}$ and (\ref{eq-time-reversal-operator}), gives
\begin{equation*}
  T\Xi_kZ_{\sigma} = - (-1)^{\#(\sigma)}\Xi_kZ_{\sigma} = - \Xi_kTZ_{\sigma}.
\end{equation*}
Since the family $\{Z_{\sigma}\mid \sigma\in \Gamma\}$ is the canonical ONB for $\mathfrak{h}$, we finally get $T\Xi_k = - \Xi_kT$.
\end{proof}

In the classical theory of time-reversal symmetry, the time-reversal symmetry of a quantum system is described by an anti-unitary operator.
The next theorem, however, shows that the time-reversal symmetry of the walk $A_w$ can be described directly by a unitary operator.

\begin{theorem}\label{thr-time-reversal-symmetry}
The unitary operator $T$ is a time-reversal operator of the walk $A_w$, namely it holds that
\begin{equation}\label{eq}
  Te^{-\mathrm{i}tA_w} = e^{\mathrm{i}tA_w}T,\quad \forall\, t\in \mathbb{R}.
\end{equation}
\end{theorem}

\begin{proof}
It follows from Proposition~\ref{prop-shift-time-commutation} that $TA_w = - A_wT$.
Thus, for a general $n\in \mathbb{N}$, we further have
\begin{equation}\label{eq}
  TA_w^n = (-1)^nA_w^nT.
\end{equation}
Now let $t\in \mathbb{R}$ be given. Then, using the boundedness of $A_w$ and the unitary property of $T$, we get
\begin{equation*}
  Te^{-\mathrm{i}tA_w}
     = \sum_{n=0}^{\infty}\frac{(-\mathrm{i}t)^n}{n!}TA_w^n
     = \sum_{n=0}^{\infty}\frac{(-\mathrm{i}t)^n}{n!}(-1)^nA_w^nT
     = e^{\mathrm{i}tA_w}T.
\end{equation*}
This is the desired.
\end{proof}

Recall that the probability distribution $P_t(\cdot \mid \xi_0)$ of the walk $A_w$ at time $t\in \mathbb{R}$ is given by
\begin{equation}\label{eq}
  P_t(\sigma\! \mid\! \xi_0) = |\langle Z_{\sigma}, e^{-\mathrm{i}tA_w}\xi_0\rangle|^2,\quad \sigma\in \Gamma,
\end{equation}
where $\xi_0$ is the initial state, which, by definition, is a unit vector in $\mathfrak{h}$.
The following proposition just shows that the probability distributions of the walk $A_w$ have some kind of time-reversal symmetry.

\begin{proposition}\label{prop-prob-dist-basic}
Let $\xi_0\in \mathfrak{h}$ be a unit vector in $\mathfrak{h}$. Then, for all $\sigma\in \Gamma$, one has
\begin{equation}\label{eq-prob-dist-basic}
P_t(\sigma\!\mid\! T\xi_0) = P_{-t}(\sigma\!\mid\! \xi_0), \quad t\in \mathbb{R}.
\end{equation}
\end{proposition}

\begin{proof}
Let $\sigma\in \Gamma$ and $t\in \mathbb{R}$ be given. By Theorem~\ref{thr-time-reversal-symmetry}
and the self-adjoint property of $T$, we have
\begin{equation*}
  P_t(\sigma\!\mid\! T\xi_0)
  = |\langle Z_{\sigma}, e^{-\mathrm{i}tA_w}T\xi_0\rangle|^2
  = |\langle Z_{\sigma}, Te^{\mathrm{i}tA_w}\xi_0\rangle|^2
  = |\langle TZ_{\sigma}, e^{\mathrm{i}tA_w}\xi_0\rangle|^2,
\end{equation*}
which together with $TZ_{\sigma} = (-1)^{\#(\sigma)}Z_{\sigma}$ gives
\begin{equation*}
  P_t(\sigma\!\mid\! T\xi_0) = \big|\big\langle (-1)^{\#(\sigma)}Z_{\sigma}, e^{\mathrm{i}tA_w}\xi_0\big\rangle\big|^2
  = |\langle Z_{\sigma}, e^{\mathrm{i}tA_w}\xi_0\rangle|^2
  = P_{-t}(\sigma\!\mid\! \xi_0).
\end{equation*}
Therefore (\ref{eq-prob-dist-basic}) is true.
\end{proof}

The next theorem further reveals that the probability distributions of the walk $A_w$ have the usual symmetry with respect to time $t$
provided its initial state $\xi_0$ meets some mild requirements.

\begin{theorem}\label{thr-distribution-symmetry-usual}
Let the initial state $\xi_0$ of the walk $A_w$ be such that $\xi_0\in \mathfrak{h}^{(+)}$ or $\xi_0\in \mathfrak{h}^{(-)}$.
Then, for all $\sigma\in \Gamma$, it holds that
\begin{equation}\label{eq}
P_t(\sigma\!\mid \!\xi_0) = P_{-t}(\sigma\!\mid \!\xi_0), \quad t\in \mathbb{R}.
\end{equation}
\end{theorem}

\begin{proof}
Let $\sigma\in \Gamma$ and $t\in \mathbb{R}$ be given. If $\xi_0\in \mathfrak{h}^{(+)}$, then, by Proposition~\ref{prop-4-2},
$\xi_0 = T\xi_0$, which together with Proposition~\ref{prop-prob-dist-basic} yields
\begin{equation*}
  P_t(\sigma\!\mid \!\xi_0) = P_t(\sigma\!\mid \!T\xi_0)= P_{-t}(\sigma\!\mid \!\xi_0).
\end{equation*}
In the case of $\xi_0\in \mathfrak{h}^{(-)}$, we can similarly get
$P_t(\sigma\!\mid \!\xi_0)= P_{-t}(\sigma\!\mid \!-\xi_0)$, which, together with the equality
$P_{-t}(\sigma\!\mid \!-\xi_0)=P_{-t}(\sigma\!\mid \!\xi_0)$, implies that $P_t(\sigma\!\mid \!\xi_0)= P_{-t}(\sigma\!\mid \!\xi_0)$.
\end{proof}

\end{document}